\documentclass[conference]{IEEEtran}
\usepackage[utf8]{inputenc}
\usepackage{cite}
\usepackage{amsmath,amssymb,amsfonts, amsthm}
\usepackage{dsfont}
\usepackage{mathtools}
\usepackage{mgraphics}
\usepackage{mdefs}
\usepackage{units}
\usepackage[flushleft]{threeparttable}
\usepackage{enumitem}

\usepackage[noend]{algorithmic}
\usepackage[ruled,linesnumbered]{algorithm2e}
\usepackage{scrextend}

\usepackage[subrefformat=parens,labelformat=parens,caption=false,font=footnotesize]{subfig}

\usepackage{tikz}
\usetikzlibrary{shapes,arrows}
\usetikzlibrary{calc}

\newtheorem{proposition}{Proposition} 
\newtheorem{remark}{Remark} 

\def\tablescale{0.85}

\newcommand\blfootnote[1]{%
  \begingroup
  \renewcommand\thefootnote{}\footnote{#1}%
  \addtocounter{footnote}{-1}%
  \endgroup
}

\title{Learning Hierarchical Resource Allocation and Multi-agent Coordination of 5G mobile IAB Nodes}

\newcommand{\titleheader}{This work has been accepted for publication in 2023 IEEE International Conference on Communications (ICC): SAC Backhaul/Fronthaul Networking and Communications Track}

\makeatletter
\def\ps@headings{%
\def\@oddhead{\mbox{}\scriptsize \titleheader}
\def\@oddfoot{\scriptsize \@date \hfil }%
}

\def\ps@IEEEtitlepagestyle{%
\def\@oddhead{\mbox{}\scriptsize \titleheader \rightmark \hfil }%
}
\makeatother

\author{\IEEEauthorblockN{Mohamed Sana, Benoit Miscopein}
\IEEEauthorblockA{CEA-Leti, Université Grenoble Alpes, F-38000 Grenoble, France\\
Email: \{mohamed.sana, benoit.miscopein\}@cea.fr}}

\begin{document}

\maketitle

\begin{abstract}
    We consider a dynamic millimeter-wave network with integrated access and backhaul, where mobile relay nodes move to auto-reconfigure the wireless backhaul. Specifically, we focus on \emph{in-band relaying} networks, which conduct access and backhaul links on the same frequency band with severe constraints on co-channel interference. In this context, we jointly study the complex problem of dynamic relay node positioning, user association, and backhaul capacity allocation. To address this problem, with limited complexity, we adopt a hierarchical multi-agent reinforcement with a two-level structure. A high-level policy dynamically coordinates mobile relay nodes, defining the backhaul configuration for a low-level policy, which jointly assigns user equipment to each relay and allocates the backhaul capacity accordingly. The resulting solution automatically adapts the access and backhaul network to changes in the number of users, the traffic distribution, and the variations of the channels. Numerical results show the effectiveness of our proposed solution in terms of convergence of the hierarchical learning procedure. It also provides a significant backhaul capacity and network sum-rate increase (up to $3.5\times$) compared to baseline approaches.
\end{abstract}

\blfootnote{This work was supported by the French government under the Recovery Plan (CRIIOT Project) and the H2020 Project DEDICAT 6G (no. 101016499).}

\vspace{-0.3cm}
\section{Introduction}
Enhanced mobile broadband services (eMBB) with high data throughput requirement (up to $20\Gbps$ peak) is one of the main targets of the recently standardized 5G networks \cite{TR38.913}. To boost the network capacity, 5G adopts cell densification together with millimeter wave (mmWave) communications to benefit from the large spectrum available at these frequencies \cite{TR38.913}. In addition, spatial reuse of the spectrum across a geographical area allows cell densification to considerably improve the coverage quality of mmWave base stations (BSs) and the performance of cell-edge users (UEs) \cite{LopezPerez2011HetNets}. However, densification poses serious challenges to radio resource management (RRM), which become complex with increasing number of UEs and BSs. Also, network capacity does not increase systematically with cell densification due to co-channel interference and limited backhaul capacity, which needs to be increased accordingly. Yet, the deployment of backhaul networks, generally relying on wired optical fibers or microwave links, is expensive, setting constraints on the backhaul capacity, which may affect network spectral efficiency and quality of service (QoS) of end-users. 

To address this problem, 5G also introduces integrated access and backhaul (IAB) networks as a cost-effective alternative to wired backhaul networks \cite{TR38.874, Weiler2014}. Indeed, the large spectrum resource available at mmWave frequencies allows partitioning the total bandwidth into parts dedicated to wireless access and backhaul networks, respectively \cite{TR38.874}. In the considered system model, multiple mobile relay stations (hereafter referred to as mIAB nodes) dynamically move to form a wireless backhaul network with an overlaid mmWave station (hereafter referred to as IAB donor), jointly providing access to multiple deployed mobile UEs. In this study, we focus on \emph{in-band relaying} IAB networks, which simultaneously conduct the access and backhaul links on the same frequency band. Such networks have stringent interference constraints as the access and backhaul links mutually interfere with each other. Therefore, dynamic coordination of mIAB nodes together with efficient joint RRM on the access and backhaul network is required. 
This problem has received a wide attention from academia and industry \cite{Domenico2013BackHaul, Michele2018, Emmanouil2018, Dingwen2018}. In \cite{Domenico2013BackHaul}, the authors propose a centralized algorithm, which optimizes the user association taking into account the load of the backhaul network. \cite{Valente2019} proposed a similar approach leveraging a Q-learning algorithm. In \cite{Wanlu2020}, the authors proposed a deep reinforcement learning (RL) algorithm for spectrum allocation, focusing on \emph{out-of-band relaying} IAB networks. Authors in \cite{Bibo2021a} propose a multi-agent RL (MARL) based algorithm to address UE's mobility. However, none of these works consider mobile IAB nodes. In contrast, we address the problem of the dynamic positioning of mIAB nodes to reconfigure the backhaul network together with resource allocation, namely the user association on the access network and backhaul capacity allocation. Moreover, we take into account environment dynamics such as interference, the mobility of UEs, the load of IAB nodes, and the variations of UEs traffic requests with time, which further make it difficult to find accurate and tractable solutions. For instance, \cite{Fouda2018} employs an exhaustive search algorithm, intractable in practice, to optimize mIAB networks with unmanned aerial vehicles. 

To solve this problem with limited complexity, we adopt a \emph{hierarchical} RL (hRL) based approach \cite{Sutton1998}. The hierarchy follows by decomposing the aforementioned problem into two sub-problems. A high-level optimization consists in dynamically coordinating mIAB nodes and defining their positioning to jointly maximize backhaul capacity and UE coverage. Then, a low-level optimization jointly determines the optimal resource allocation on the access and backhaul network. We address these two optimizations within a MARL framework, where we model each mIAB node and each UE as an independent agent relying solely on a few observations of the radio environment to manage radio resources. We introduce new reward functions specifically designed to learn hierarchical policies, with limited complexity. The proposed solution is flexible by design, scalable, and automatically re-configures the backhaul network with respect to (\textit{w.r.t.}) environment dynamics and the requirements on the access network. 

\section{System Model}\label{sec:sys-model}

\begin{figure}[!t]
    \centering
    \includegraphics[width=0.97\columnwidth]{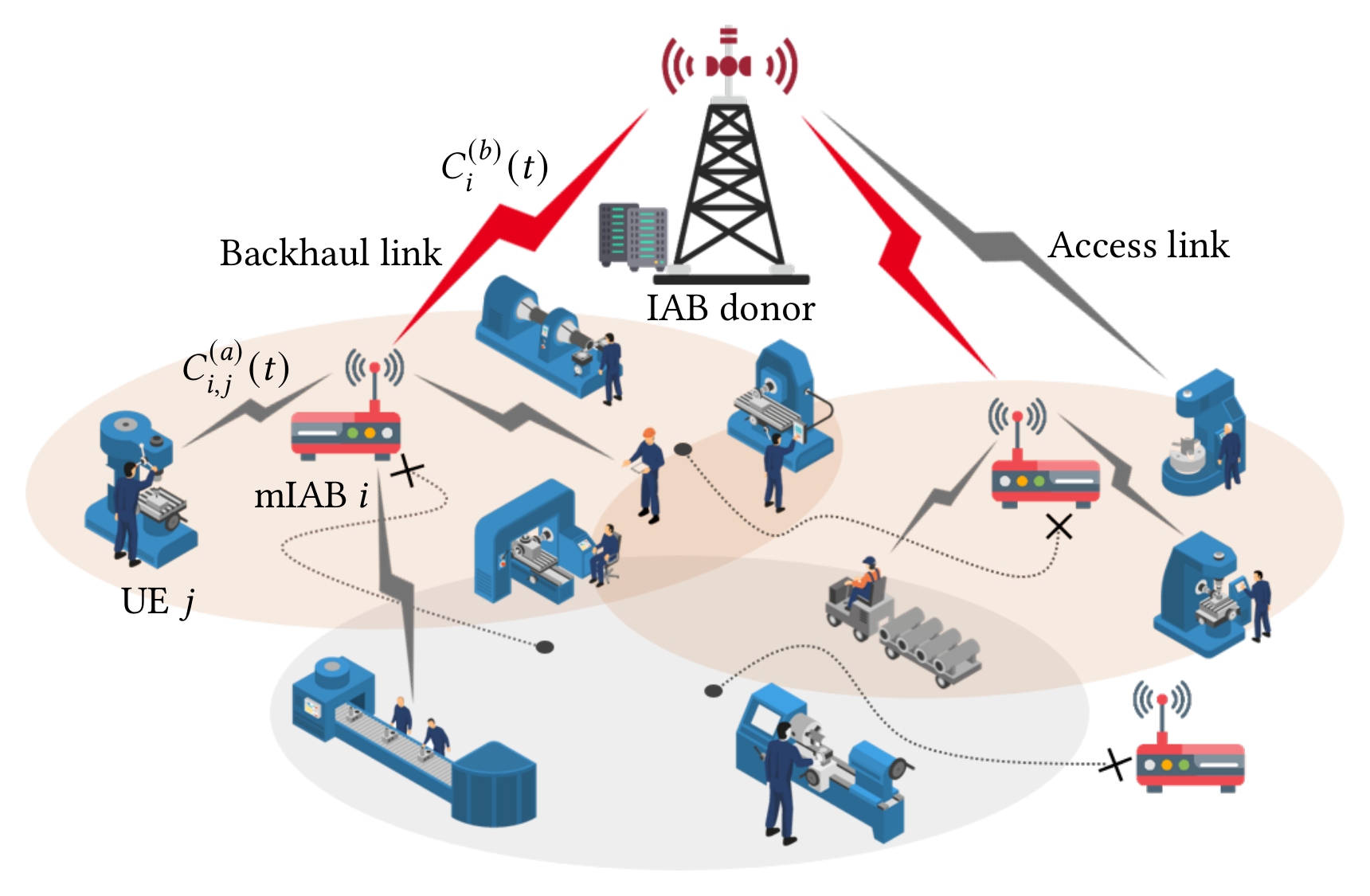}
    \caption{System model with 3 IAB nodes deployed with 1 IAB donor}
    \label{fig:sys-model}
\end{figure}

We consider a downlink mmWave network, as in \fig{fig:sys-model}, composed of one fixed IAB donor forming a backhaul network with $N_s$ mIAB nodes, which move in the network to provide access to a set $\mathcal{U}(t) = \{1,\dots, K(t)\}$ of $K(t)$ UEs at time $t$. We denote with $\mathcal{S}_0 = \{1, \dots, N_s\}$ the set of mIAB nodes and with $\mathcal{S} = \mathcal{S}_0 \cup \{0\}$ the set of all IAB stations, including the IAB donor indexed by $0$. In this dynamic network, each mIAB node $i\in\mathcal{S}_0$, equipped with a receiving antenna can adapt its 2D location $\ell_i(t)$ in a region $\mathcal{L}\subset\mathbb{R}^2$ of space, to set up a wireless backhaul link with the IAB donor. Once positioned, mIAB $i$ can simultaneously serve at most $L_i$ UEs due to limited beamforming capabilities. We use $L_0$ and $M$ to denote the maximum number of downlink access and backhaul links of the IAB donor, respectively. Eventually, given $\rho_i(t)$, refers to as the load of (number of active UEs of) IAB station $i\in\mathcal{S}$ at time $t$, we define $\kappa(t)$ to indicate beam coverage rate (\ien, total beam utilization) of access network over time:
\begin{align}
    \kappa(t) = \frac{\sum_{i\in\mathcal{S}} \rho_i(t)}{\minof{K(t)}{\sum_{i\in\mathcal{S}} L_i}},~
\end{align}

\subsection{Access and backhaul capacity}

As we focus on \emph{in-band} relaying IAB network, the access and the backhaul links partially overlap in frequency. We assume spatial division multiple access (SDMA) for backhaul and access network. In our system model, the IAB donor allocates all the available mmWave bandwidth $B$ to its served mIAB nodes, which in turn, operate in $\mu$-duplex mode \cite{AlAmmouri2016}, conducting all the access links only on a portion $\mu B$ of the band, where $\mu\in[0,1]$.
In this context, interference results from the overlapping of operating beams, as we do not specifically optimize beamformers. Thus, when UE $j$ is receiving data from IAB station $i$, it experiences a downlink signal-to-interference-plus-noise ratio $\mathrm{SINR}_{i,j}^{(a)}$, whose expression reads as:
\begin{align}
    \mathrm{SINR}_{i,j}^{(a)}(t) = \frac{\chi_{i,j}(t) P_{i,j}^{\rm Tx} G_{i,j}^{\rm Tx}(t) G_{i,j}^{\rm H}(t) G_{i,j}^{\rm Rx}(t) }{I_{i,j}^{(a)}(t) + \mu N_0B}.
\end{align}
Here, $P_{i,j}^{\mathrm{Tx}}$ is the transmit power from IAB station $i$ towards UE $j$, $N_0$ is the Gaussian noise power spectrum density, and $G_{i,j}^{\mathrm{Tx}}(t)$ and $G_{i,j}^{\mathrm{Rx}}(t)$ are the transmit and receive antenna gain between IAB station $i$ and UE $j$, respectively. Also, $\chi_{i,j}(t)$ denotes the small-scale fading coefficient, and $G_{i,j}^{\mathrm{H}}(t)$ is the channel gain, which captures the distance-dependent path loss and the large-scale shadowing effect. In particular, note that $G_{i,j}^{\mathrm{H}}(t)$ is affected by the location $\ell_i(t)$ of mIAB node $i$. Eventually, as we assume full spatial reuse of the frequency across all mmWave links, the \emph{total} interference $I_{i,j}^{(a)}(t)$, experienced by UE $j$ communicating with IAB station $i$, results from the contribution of access links intra- and inter-cell interference and the backhaul links inter-cell interference\footnote{{The mutual interference between access and backhaul links sums only over the overlapping bands. The corresponding amplitude factor is $\sqrt{\mu}$, assuming a rectangular pulse shape (in frequency domain) for access and backhaul signal (see \cite{AlAmmouri2016}). Although interesting, the optimization of $\mu$ is out of the scope of the present paper and will be addressed in future work.}}: 
\begin{align}\label{eq:access-interf}
    &I_{i,j}^{(a)}(t) = \chi_{i,j}(t)\sum_{\mathclap{j'\in\mathcal{U}\backslash\{j\}}} x_{i,j'}(t) P_{i,j'}^{\rm Tx} G_{(i,j')\rightarrow j}^{\rm Tx}(t) G_{i,j}^{\rm H}(t) G_{i,j}^{\rm Rx}(t)\nonumber\\
    &+\sum_{\mathclap{i'\in\mathcal{S}\backslash\{i\}}}~~\quad\sum_{\mathclap{j'\in\mathcal{U}\backslash\{j\}}} x_{i',j'}(t)\chi_{i',j}(t) P_{i',j'}^{\rm Tx} G_{(i',j')\rightarrow j}^{\rm Tx} G_{i',j}^{\rm H}(t) G_{(i,j)\leftarrow i'}^{\rm Rx}\nonumber\\
    &+\chi_{0,j}(t){\mu}\sum_{\mathclap{i'\in\mathcal{S}_0}} z_{i'}(t) P_{0,i'}^{\rm Tx} G_{(0,i')\rightarrow j}^{\rm Tx} G_{0,j}^{\rm H}(t) G_{(i,j)\leftarrow 0}^{\rm Rx}(t).
\end{align}
Here, $x_{i,j}(t)$ is the binary user association variable, which equals $1$ if UE $j$ is associated with IAB node $i$, and $0$ otherwise. Similarly, $z_i(t)$ is the backhaul link association variable, which equals $1$ if the backhaul link between IAB node $i$ and the IAB donor is active, and $0$ otherwise; $G_{(u, v)\rightarrow w}^{\rm Tx}$ is the transmit antenna gain from $u$ towards $w$, when $v$ is receiving data from $u$, $G_{(u, v)\leftarrow w}^{\rm Rx}$ is the receive antenna gain from $w$ towards $v$, when $v$ is receiving data from $u$. Hence, the access link's capacity between IAB station $i$ and UE $j$, denoted $C_{i,j}^{(a)}(t)$ is:
\begin{align}
    C_{i,j}^{(a)}(t)={\mu B}\cdot\mathrm{log}_2\left(1 + x_{i,j}(t) \mathrm{SINR}_{i,j}^{(a)}(t)\right).
\end{align}

Similarly, the \emph{total} capacity $C_i^{(b)}(t)$ of the backhaul link between mIAB node $i\in\mathcal{S}_0$ and IAB donor reads as: 
\begin{align}
    C_{i}^{(b)}(t)={B} \cdot\mathrm{log}_2\left(1 + z_{i}(t)\mathrm{SINR}_{i}^{(b)}(t)\right),
\end{align}
where the experienced $\mathrm{SINR}_{i}^{(b)}(t),~\forall i\in\mathcal{S}_0$ is given as follows:
\begin{align}
    \mathrm{SINR}_{i}^{(b)}(t) = \frac{\chi_{0,i}(t) P_{0,i}^{\rm Tx} G_{0,i}^{\rm Tx}(t) G_{0,i}^{\rm H}(t) G_{0,i}^{\rm Rx}(t) }{I_{i}^{(b)}(t) + {N_0B}}.
\end{align}
Here, and in contrast to Eq. \eqref{eq:access-interf}, $I_{i}^{(b)}(t)$ comprises the intra-cell backhaul links interference, the access-links interference and the self-interference\footnotemark[1] resulting from the simultaneous transmission and reception of the corresponding mIAB node:
\begin{align}\label{eq:backhaul-interf}
    &I_{i}^{(b)}(t) =\chi_{0,i}(t)\sum_{\mathclap{i'\in\mathcal{S}\backslash\{i\}}} z_{i'}(t) P_{0,i'}^{\rm Tx} G_{(0,i')\rightarrow i}^{\rm Tx}(t) G_{0,i}^{\rm H}(t) G_{0,i}^{\rm Rx}(t)\nonumber\\
    &+{\mu}\sum_{\mathclap{i'\in\mathcal{S}\backslash\{i\}}} \quad \sum_{\mathclap{j'\in\mathcal{U}}} x_{i',j'}(t) \chi_{i',i}(t) P_{i',j'}^{\rm Tx} G_{(i',j')\rightarrow i}^{\rm Tx} G_{i',i}^{\rm H}(t) G_{(0,i)\leftarrow i'}^{\rm Rx}\nonumber\\
    &+{\mu}\sum_{j'\in\mathcal{U}} z_i(t) x_{i,j'}(t)P_{i,j'}^{\rm Tx} \xi_i,
\end{align}

where, $\xi_i$ is the self-interference cancellation gain of mIAB node $i$, which is the ratio of the received self-interference power after and before the interference cancellation.

\subsection{User effective rate and network sum-rate}
Let $D_j(t)$ denote the traffic request of UE $j$ at time $t$ (in$\bps$). From the backhaul viewpoint, $T_{i,j}(t) = \minof{D_j(t)}{C_{i,j}^{(a)}(t)}$ can be viewed as the effective data requirement on the access link $i\rightarrow j$. We assume that UEs on the same backhaul link $i\in\mathcal{S}_0$ share its capacity $C_i^{(b)}(t)$. Hence, let $\beta_{i,j}(t)\in[0,1]$ denote the fraction of the backhaul capacity dedicated to UE $j$ communicating with mIAB node $i$ at time $t$. Thus, $\forall j\in\mathcal{U}(t)$, the instantaneous \emph{effective rate} $R_{i,j}(t)$ perceived by UE $j$ from IAB station $i$ reads as:
\begin{align}\label{eq:eff-rate}
        R_{i,j}(t) = \begin{dcases}	                                     \minof{T_{i,j}(t)}{\beta_{i,j}(t) z_i(t) C_i^{(b)}(t)},\forall i \in \mathcal{S}_0,\\
				T_{i,j}(t),\hspace{3.95cm}\text{if $i=0$}.
				    \end{dcases}
\end{align}
Eventually, the total network sum-rate $R(t)$ reads as:
\begin{align}\label{eq:objective}
    R(t) = \sum_{i\in\mathcal{S}}\sum_{j\in\mathcal{U}(t)} R_{i,j}(t),
\end{align}

\section{Problem Formulation} \label{sec:problem}
Our goal is to maximize long-term network sum-rate by jointly addressing four sub-problems, which consist in i) finding the optimal locations of mIAB nodes, ii) determining the active backhaul links, iii) associating users \wrt the selected backhaul links, and iv) optimally allocating backhaul capacity to each served UE under long-term and instantaneous constraints:
\begin{align}
	\underset{\mathrmbold{\Psi}(t)}{\mathrm{max}}~~& \lim_{T\rightarrow+\infty}\frac{1}{T}\sum_{t=1}^T R(t), \tag{$\mathcal{P}_0$} \label{eq:P0}\\[0cm]
    	\mathrm{s.t.~}~ & {} \lim_{T\rightarrow+\infty}\frac{1}{T}\sum_{t=1}^T\kappa(t) \geq \kappa_0, \tag{$\mathcal{C}_1$} \label{eq:C1}\\
	    {}&x_{i,j}(t) \in \{0,1\}, & \forall i\in \mathcal{S}, j\in \mathcal{U}(t), \tag{$\mathcal{C}_2$} \label{eq:C2}\\
		{}&\sum_{j \in \mathcal{U}(t)}x_{i,j}(t) \leq L_{i}, & \forall i \in \mathcal{S}, \tag{$\mathcal{C}_3$} \label{eq:C3}\\
		{}&\sum_{i \in \mathcal{S}}x_{i,j}(t) \leq 1, & \forall j \in \mathcal{U}(t), \tag{$\mathcal{C}_4$} \label{eq:C4}\\
        {}& z_i(t) \in \{0,1\}, & \forall i\in\mathcal{S}_0, \tag{$\mathcal{C}_5$} \label{eq:C5}\\
		{}&\sum_{\mathclap{i \in \mathcal{S}_0}}z_i(t)\leq M, & \tag{$\mathcal{C}_6$} \label{eq:C6}\\
    	&\beta_{i,j}(t) \in [0, 1], &  \forall i\in \mathcal{S}_0,j\in \mathcal{U}(t), \tag{$\mathcal{C}_7$} \label{eq:C7}\\
    	{}&\sum_{j \in \mathcal{U}(t)}\beta_{i,j}(t) \leq 1, & \forall i \in \mathcal{S}_0, \tag{$\mathcal{C}_8$} \label{eq:C8}\\
    	{}&\ell_i(t)\in\mathcal{L}\subset\mathbb{R}^2, & \forall i \in \mathcal{S}_0, \tag{$\mathcal{C}_9$} \label{eq:C9}\\
    	{}&\norm{\ell_i(t+1)-\ell_i(t)} \leq \Delta \ell, & \forall i \in \mathcal{S}_0, \tag{$\mathcal{C}_{10}$} \label{eq:C10}
\end{align}
where $\mathrmbold{\Psi}(t) = \{x_{i,j}(t), z_i(t), \beta_{i,j}(t), \ell_i(t), \forall i,j\}$ and the expectation in \eqref{eq:P0} is taken \wrt the random traffic requests and channels realization, whose statistics are unknown. Here, constraint \eqref{eq:C1} guarantees a minimum $\kappa_0$ of long-term beam utilization; \eqref{eq:C2}-\eqref{eq:C4} constrain each IAB station $i$ to serve at most $L_i$ UEs simultaneously and each UE to be associated with only one IAB station at a time. In addition, \eqref{eq:C5}-\eqref{eq:C6} ensure that at most $M$ backhaul links are active simultaneously. Also, \eqref{eq:C7}-\eqref{eq:C8} guarantee that the fractions $\beta_{i,j}(t)$ of the backhaul capacity allocated to UEs on the same backhaul are positive and sum to at most one at each time $t$. Eventually, \eqref{eq:C9}-\eqref{eq:C10} ensure that mIAB nodes move in a region $\mathcal{L}$ of space by no more than $\Delta \ell$ meters at a time. Note that the locations $\ell_i(t)$ of mIAB nodes affect path losses, thus channel gains and $R(t)$.

Problem \eqref{eq:P0} is a combinatorial non-convex optimization, whose complexity grows exponentially \wrt the numbers of UEs, thus intractable with conventional optimization tools \cite{sana2020UA}. In addition, the solution of aforementioned sub-problems mutually affects each other. For instance, the optimal allocation of backhaul capacity  depends on the user association, which in turn depends on the location of mIAB nodes due to the mutual interference between access and backhaul links.

\begin{proposition} \label{prop:pure-association}
Given a deployment of mIAB nodes, if the optimal user association is known, then, 
\begin{enumerate}
\item the optimal association of backhaul links is given by:
\begin{align}\label{eq:active_backahul}
    z_i^*(t) = \mathds{1}\left(\rho_i(t) > 0\right), \quad \forall i\in\mathcal{S}_0,
\end{align}
where $\rho_i(t)=\sum_{j \in \mathcal{U}(t)}x_{i,j}^*(t)$ is the load of mIAB node $i$. Here, $\mathds{1}(\mathrm{cond}(x))$ is the indicator function, which equals $1$ if $\mathrm{cond}(x)$ is satisfied and $0$ otherwise, and

\item the optimal allocation of backhaul capacity is obtained by solving the following convex problem:
\begin{align}
	\underset{\{\beta_{i,j}\}_{i,j}}{\mathrm{max}}~ &\sum_{i\in\mathcal{S}_0^+}\sum_{j\in\mathcal{U}^+} \minof{T_{i,j}(t)}{\beta_{i,j}(t) C_i^{(b)}(t)}, \tag{$\mathcal{P}_1$} \label{eq:BH-BW}\\
	\mathrm{s.t.~}~ & \eqref{eq:C7}~\text{and}~\eqref{eq:C8} \nonumber,
\end{align}
where $\mathcal{S}_0^+$ denotes the set of mIAB nodes with active backhaul links and $\mathcal{U}^+$ is the set of active UEs.
\end{enumerate}
Thus, \eqref{eq:P0} can be reduced to i) a mIAB nodes coordination problem, consisting in determining the optimal positioning of the mIABs at each time and ii) a user association problem, consisting in determining the optimal assignment of UEs and mIABs to maximize long-term network sum-rate.
\end{proposition}

\begin{proof}[Sketch of proof]
The proof follows by first noting that there is no need to activate a backhaul link if no UE has requested a connection to the corresponding mIAB node. Conversely, if a backhaul link is not active, there is no need to associate UEs to the corresponding mIAB node. The user association is not only optimal from the access network viewpoint but also from the backhaul perspective {(guaranteeing at the same time constraint \eqref{eq:C6})}. Thus, if we know the optimal user association, then the optimal association of backhaul links can be immediately deduced using Eq. \eqref{eq:active_backahul}. In this case, given a deployment of mIAB nodes, \eqref{eq:P0} reduces to \eqref{eq:BH-BW}: an optimization problem over only $\{\beta_{i,j}(t),~\forall i,j\}$. Then, the convexity of \eqref{eq:BH-BW} follows by observing that constraints \eqref{eq:C7} and \eqref{eq:C8} are convex. In addition, $C_i^{(b)}(t)$ and $T_{i,j}(t)$ are constant, positive, and independent of $\beta_{i,j}(t), ~\forall i,j$, so that $\minof{T_{i,j}(t)}{\beta_{i,j}(t) C_i^{(b)}(t)}$ is a concave function \wrt $\beta_{i,j}(t)$, which concludes the proof. 
\end{proof}

\begin{remark}
As we assume full spatial reuse of the frequency across all mmWave links, variables $\beta_{i,j}(t)$ are decoupled in \eqref{eq:BH-BW}, which can be solved distributively at each mIAB node.
\end{remark}

Following Proposition \ref{prop:pure-association}, our goal is now to i) find the optimal positioning of the mIABs at each time and ii) the optimal user association from the access and backhaul point of view. Jointly solving these two sub-problems remains complex because of their combinatorial and non-convexity nature. The optimal solution becomes even more challenging when considering network dynamics, including size-variable topology, shadowing, fading, and UEs mobility. To limit such complexity, we adopt the following hierarchical MARL approach.

\section{Proposed Solution via hierarchical MARL}\label{sec:solution}

\subsection{General hierarchical optimization framework}
To solve Problem \eqref{eq:P0}, we extend the standard reinforcement learning framework to a hierarchical two-level structure, where a high-level policy $\pi^{(h)}$ (backhaul link manager) defines a strategy for positioning mIAB nodes, setting goals for a low-level policy $\pi^{(l)}$ (access link manager), which determines the user association strategy. More specifically, high-level policy coordinates mIAB nodes, dynamically adapting their locations to jointly maximize \emph{user coverage} (by at least one mIAB node) and \emph{sum backhaul capacity} ${C}^{(b)}(t)=\sum_{i\in\mathcal{S}_0} C_i^{(b)}(t)$. Then, low-level policy determines the user association to maximize \emph{network sum-rate} \eqref{eq:objective} while guaranteeing long-term \emph{beam utilization} \eqref{eq:C1}. We cast these two optimization problems into separate MARL frameworks, modelling each mIAB and UE as autonomous agent, respectively learning the high- and low-level policy through their interaction with the radio environment. At each time $t$, the radio environment provides to each mIAB node $i$ and to each UE $j$, a high- and low-level observation $\obs{i}{h}{t}\in\mathbb{R}^h$ and $\obs{j}{l}{t}\in\mathbb{R}^l$, respectively. In our setting, the observation of mIAB node $i$ coincides with its current location; thus, $\obs{i}{h}{t}=\ell_i(t)$. In contrast, we define UE $j$'s observation as $\obs{j}{l}{t}=\{D_j(t), \mathrm{RSS}_i(t), \mathrm{AoA}_i(t), R_{i,j}(t-1)\}_{i\in\mathcal{S}}$, which comprises its traffic request, instantaneous local signal measurements such as received signal strength (RSS), corresponding estimated angle of arrival (AoA), and previously experienced throughput \wrt different IAB stations, similar to \cite{sana2020UA}. In addition to its local observation, we assume that each entity $e$, either mIAB node or UE, can build a local radio map $\boldsymbol{\phi}_e(t)\in\mathbb{R}^{p\times n}$ capturing its relative perception of the surrounding radio environment. Combined with local observations, this map allows for effective learning of hierarchical policies.  

\begin{figure}[!t]
    \centering
    \includegraphics[width=0.97\columnwidth]{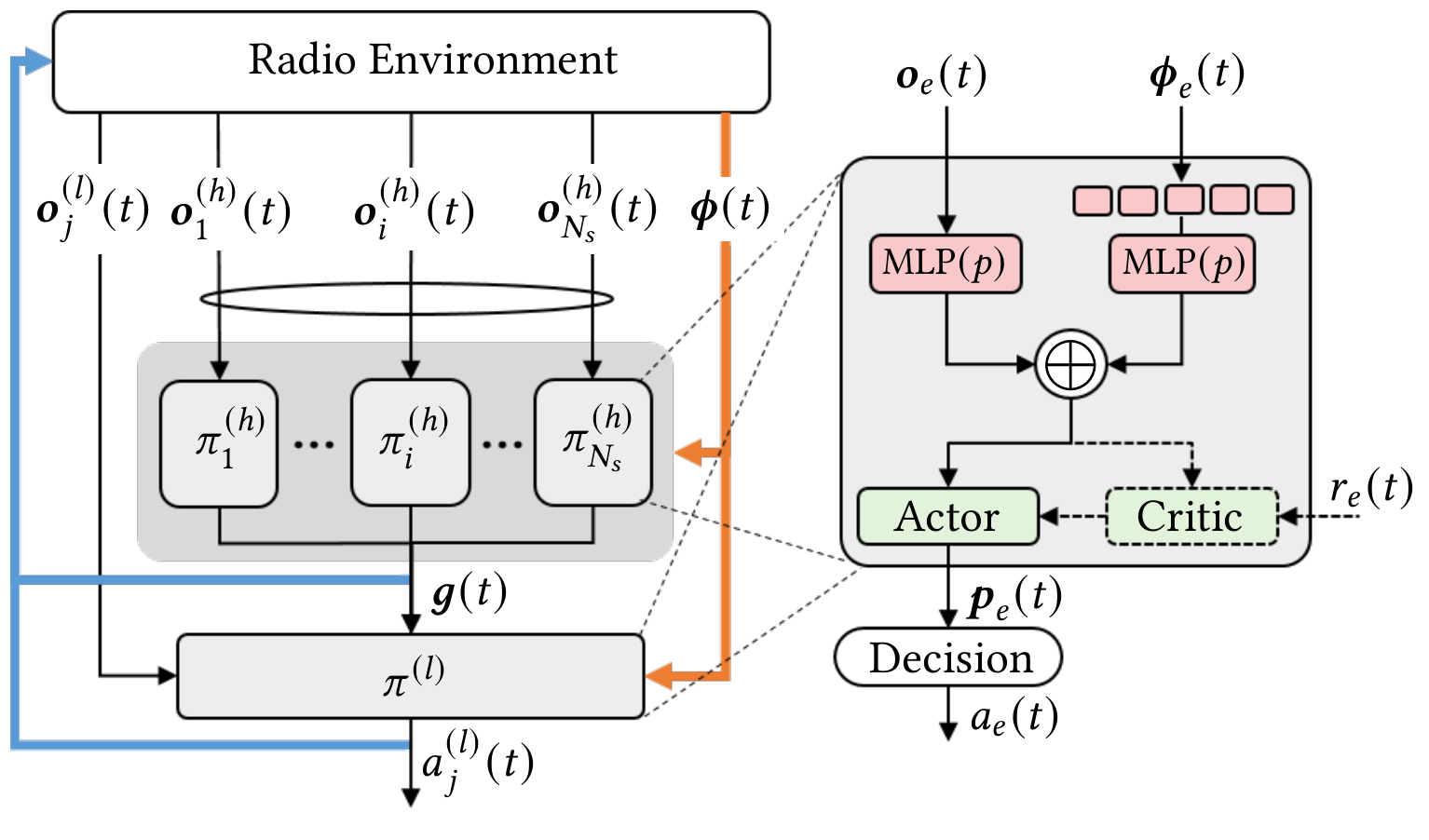}
    \caption{Hierarchical policy network architecture.} 
    \label{fig:policy_archi}
\end{figure}

\subsection{Learning local radio map}
In our framework, each entity learns to generate a local radio map $\boldsymbol{\phi}_e(t)\in\mathbb{R}^{p\times n}$ from locations information $\left\{\ell_{e'}(t), ~e'\in\mathcal{N}_e(t)\right\}$ signaled by neighboring entities $\mathcal{N}_e(t)$ to limit complexity. In this dynamic network where the number of entities and their position change over time, the size of $\mathcal{N}_e(t)$ and the order of received messages vary accordingly. To come out with a method for constructing $\boldsymbol{\phi}_e(t)$, which is size-invariant \wrt $\mathcal{N}_e(t)$ and permutation-invariant \wrt the received messages, we adopt idea from neural attention mechanism \cite{vaswani2017attention}. Specifically, we define $\mathrmbold{k}_{e,  e'}^{(\head)} = \params_{e,k}^{(\head)} (\ell_e(t) - \ell_{e'}(t))^T\in\mathbb{R}^n$ and $\mathrmbold{v}_{e, e'}^{(\head)} = \params_{e,v}^{(\head)} (\ell_e(t) - \ell_{e'}(t))^T\in\mathbb{R}^n$, $\forall \head=1\dots N_{\head}$, which we refer to as the relative key and value of entity $e$ \wrt entity $e'$. Here, $N_{\head}$ is the number of filters (also known as attention heads \cite{vaswani2017attention}); $\params_{e,k}^{(\head)}, \params_{e,v}^{(\head)} \in \mathbb{R}^{n\times 2}$ are learnable parameters. Similarly, let $\mathrmbold{q}_{e}^{(\head)} = \params_{e,q}^{(\head)} \ell_e(t)^T\in\mathbb{R}^n$ refers to as the query of entity $e$, where $\params_{e,q}^{(\head)}$ is also learnable parameter. Let $\mathrmbold{V}_e^{(\head)} = [\mathrmbold{v}_{e, e'}^{(\head)},~\forall e'\in\mathcal{N}_e(t)]$ denote the value matrix. We compute the attention matrix $\mathrmbold{A}_e^{(\head)}$ using dot-product mechanism \cite{vaswani2017attention}: $\mathrmbold{A}_e^{(\head)} = \softmax\left(\left[(\sqrt{n})^{-1}{\mathrmbold{q}_e^{(\head)} {\mathrmbold{k}_{e, e'}^{(\head)}}^T}, ~\forall {e'\in\mathcal{N}_e}\right]\right)$, where $\softmax(\cdot)$ is the normalized exponential function. It represents the interaction between entity $e$ and its neighbors $e'$ at the $\head$-th attention head. Finally, we compute the local radio map of entity $e$ by concatenating the outputs of attention heads, which we obtain via scalar product between $\mathrmbold{A}_e^{(\head)}$ and $\mathrmbold{V}_e^{(\head)}$:
\begin{align}
    \boldsymbol{\phi}_e(t) = \params_{e,\phi} \left[\bigoplus_{\head=1}^{N_{\head}} \mathrmbold{A}_e^{(\head)}\mathrmbold{V}_e^{(\head)} \right].
\end{align}
Here, $\params_{e,\phi}\in\mathbb{R}^{p\times N_{\head}}$ is a learnable parameter and $\oplus$ denotes the row concatenation operator.

\subsection{Learning high-level policy}\label{sec:high-level-policy} 

In our scenario, each mIAB $i$ maintain its own policy $\pi_{i}^{(h)}$. Given $\langle\obs{i}{h}{t}, \boldsymbol{\phi}_i(t)\rangle$, $\pi_{i}^{(h)}$ produces a probability distribution $\boldsymbol{p}_i^{(h)}(t) = \pi_{i}^{(h)}(\obs{i}{h}{t},{\boldsymbol{\phi}_i(t)})$ over the action space $\mathcal{A}^{(h)}$ representing the set of possible directions of movement along $x$- or $y$-axis including immobility, \ie $\card{\mathcal{A}^{(h)}}=5$. From $\boldsymbol{p}_i^{(h)}(t)$, mIAB $i$ samples high-level action $a_i^{(h)}(t)$ and moves along the selected direction with fixed step size $\Delta\ell$ to maximize the expected sum of $\gamma_h$-discounted returns $\sum_{t=0}^{T_h-1}\gamma_h^{t} r_i^{(h)}(t)$ over a time horizon $T_h$. Accordingly, we define the high-level reward to maximize user coverage and \emph{sum} backhaul capacity: 
\begin{align}\label{eq:mIAB-reward}
    r_i^{(h)}(t) = -(1-\delta_i^{(h)}(t)) {d}_i(t) - \delta_i^{(h)}(t)({d}_0 - {C}^{(b)}(t)).
\end{align}
Here, $\delta_i^{(h)}(t) = \mathds{1}(d_i(t) < d_0)$, where $d_0$ is a desirable reference distance and $d_i(t)=\norm{\ell_i(t) - \ell_i^*(t)}$ is the distance between current mIAB node $i$'s location and its optimal position $\ell_i^*(t)$. Since $\ell_i^*(t)$ is not known \emph{a priori}, we approximate it, \emph{during the training process only}, with the location of the closest centroid, linearly assign to mIAB node $i$ after clustering UEs using e.g. \texttt{Kmeans} algorithm. In this way, we push mIAB nodes towards positions, which jointly maximize user coverage (first term of \eqref{eq:mIAB-reward}) and backhaul capacity (second term of \eqref{eq:mIAB-reward}).

\subsection{Learning low-level policy}\label{sec:low-level-policy}
Unlike mIAB nodes, UEs share the same policy $\pi^{(l)}$, thus reducing complexity. Given $\langle\obs{j}{l}{t}, \boldsymbol{\phi}_j(t)\rangle$ and the goal $\boldsymbol{g}(t) = \{\ell_i(t), ~\forall i\in\mathcal{S}_0\}$ defined by the high-level policy, $\pi^{(l)}$ produces a probability vector $\boldsymbol{p}_j^{(l)}(t) = \pi^{(l)}(\obs{j}{l}{t},\boldsymbol{\phi}_j(t)\given{ \boldsymbol{g}(t)})$ over the action space $\mathcal{A}^{(l)}\subset(\mathcal{S}\cup\{\emptyset\})$ as a UE (\eg in an outage) may not be associated with any station. Based on $\boldsymbol{p}_j^{(l)}(t)$, UE $j$ samples its action $a_j^{(l)}(t)$ corresponding to either a decision to stay idle or an association request towards an IAB station to maximize expected sum of $\gamma_l$-discounted returns $\sum_{t=0}^{T_l-1}\gamma_l^{t} r_j^{(l)}(t)$ over a time horizon $T_l$. We define the goal-conditioned reward as:
\begin{align}\label{eq:user-reward}
    r_j^{(l)}(t) = (1-\delta^{(l)}(t))\kappa(t) + \delta^{(l)}(t) (\kappa_0+R(t)).
\end{align}
Similarly to \eqref{eq:mIAB-reward}, we define $\delta^{(l)}(t) = \mathds{1}(\kappa(t) \geq \kappa_0)$ so that maximizing \eqref{eq:user-reward} allows jointly optimizing beam coverage rate $\kappa(t)$ to guarantee constraint \eqref{eq:C1} (first term of \eqref{eq:user-reward}), and maximize network sum-rate $R(t)$ (second term of \eqref{eq:user-reward}).

\begin{remark}
In practice, the high-level decisions are made only every $T_l$ steps, corresponding to the backhaul update frequency, and last $T_h$ steps at the end of which the low-level policy operates to produce low-level actions. Also, for effective learning, we normalize $d_i(t)$ in \eqref{eq:mIAB-reward} \wrt to predefined maximum distance, and ${C}^{(b)}(t)$ in \eqref{eq:mIAB-reward} and $R(t)$ in \eqref{eq:user-reward} by their average values (\wrt environment randomness).
\end{remark}

\subsection{Policy architecture and learning mechanism}\label{sec:policy-architecture}
We adopt the same architecture for high- and low-level policies, briefly described in \fig{fig:policy_archi}. We first encode each entity's local observation and radio map (after being flattened) using a multi-layer perceptron (MLP) of $p$ neurons. Then, we concatenate the resulting encoding vectors, which serve as input for an actor-critic framework that we optimize end-to-end using the well-known proximal policy optimization \cite{schulman2017ppo}. We refer readers to \cite{schulman2017ppo} for a full description.

\section{Numerical results}\label{sec:result}

We randomly deploy $K_0=25$ UEs under the coverage of $N_s=3$ mIAB nodes and one IAB donor. In this dynamic network, the number and position of UEs can change with time. We adopt a random way point mobility model for UEs (velocity $\sim [0,1]~\mathrm{ms}^{-1}$), which is a standard practice \cite{mitsche2014random}. However, our proposed mechanism is, by design, agnostic to UEs mobility model. Also, we model the dynamic of UE's traffic request as a Poisson process with intensity randomly chosen between three types of service, corresponding to an average data rate demand of $5\Mbps$, $200\Mbps$, and $1.5\Gbps$. Table \ref{simu-params} summarizes simulation parameters. 
We empirically define learning parameters $p=n=128, ~N_d=8$, and compose actor and critic network with one MLP of $2p$ neurons. All layers use a rectifier linear unit activation. Also, we set the learning rate of the actor and critic to $10^{-4}$, the discounting factors $\gamma_h=0.95$, $\gamma_l=0.6$, and $T_h=T_l=250$.

\begin{table}[!t]
    \centering
    \caption{Simulations parameters \protect\cite{TR38.874}}
    \label{simu-params}
    \scalebox{\tablescale}{
    \renewcommand{\arraystretch}{1.1}
    \begin{threeparttable}
    \begin{tabular}{l|c}
       \hline
        \textbf{Parameters} & \textbf{Values} \\
       \hline
       \hline
        Carrier frequency & $28\GHz$\\
        System bandwidth, B & {$300\MHz$}\\
        Bandwidth partition, $\mu$ & {$1/3$}\\
        Thermal noise, $N_0$ & $-174\dBm/\mathrm{Hz}$\\
        Small-scale fading ($\sim \text{$m$-Nakagami}$) & $m=3$ \\
        {Large-scale fading} ($\sim\mathcal{N}(0,\sigma_0^2)$) & $\sigma_0^2=9\dB$ (users) ; $\sigma_0^2=3\dB$ (relays) \\
        {Path-loss (relay-users, donor-relays)} & $132.89 + 25\mathrm{log}_{10}(d~[\mathrm{km}])$\\
        {Path-loss (donor-users)} & $154.1 + 25\mathrm{log}_{10}(d~[\mathrm{km}])$\\
        TX power (backhaul / access) & $43\dBm$ / $33\dBm$\\
        Antenna gain (IAB nodes + donor) & {10x10 antenna array} \cite[Fig. 5, diag. 2]{sana2020UA}\\ 
        Antenna gain (users) & {$0\dBi$ (omnidirectional)}\\ 
        Beamforming & $L_i = 10~\forall i\in\mathcal{S}_0; ~L_0 = 7;~ M=3$\\
        Cell radius & $100\m$\\
        Coverage range & $50\m$ (nodes); $75\m$ (donor)\\
        Step size $\Delta\ell$ & $5\m$\\
        Target values & $\kappa_0=0.8$; $d_0=10\m$\\
        Self-interference gain $\xi$ & $-100\dBm$ \\
       \hline
    \end{tabular}
    \end{threeparttable}}
\end{table}

\begin{figure}[!t]
\vspace{-0.4cm}
    \centering
    \includegraphics[width=0.9\columnwidth]{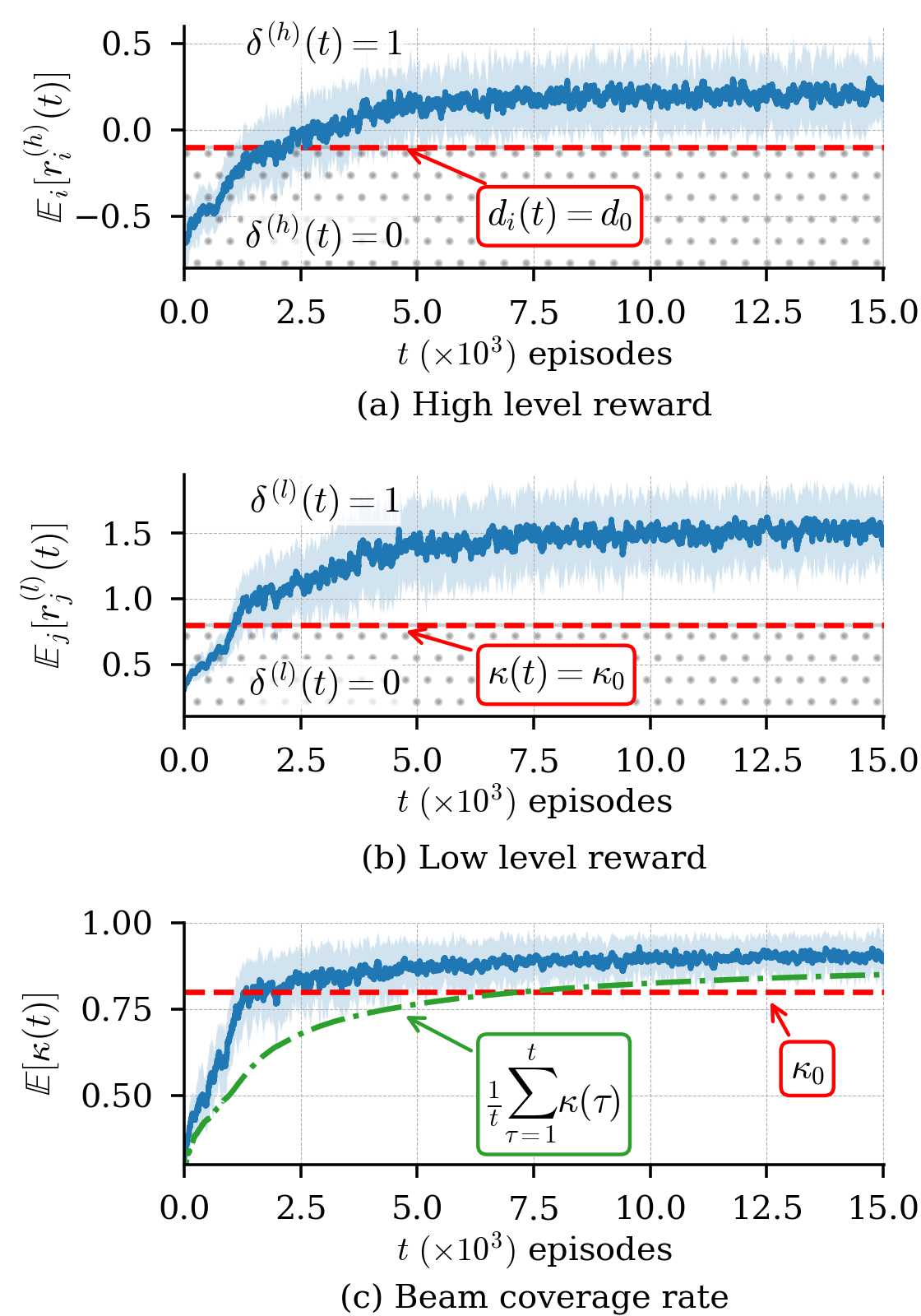}
    \caption{Hierarchical policies learning convergence.}
    \label{fig:reward_conv}
\end{figure}

\vspace{0.25cm}
\noindent
\textbf{Hierarchical Policy Convergence.} %
We first assess the convergence property of our proposed solution. \fig{fig:reward_conv} shows high-level reward, low-level reward, and beam coverage rate, respectively, over episodes of the training process. For sake of clarity, we plot the associated rolling average and standard deviation over a $100$-sized window. On \fig{fig:reward_conv}{a}, the dotted region indicates when $\delta_i^{(h)}(t)$ equals zero. In this region, mIAB agents are optimizing the user coverage (see Eq. \eqref{eq:mIAB-reward}). The reward crossing this region indicates that not only does the high-policy effectively learn to position mIAB nodes to optimize user coverage but also to optimize the backhaul capacity. So is the low-level policy on \fig{fig:reward_conv}{b}, which also learn to jointly optimize the beam coverage and network sum-rate, as the low-level reward also crosses the dotted region where the focus is on maximizing $\kappa(t)$ only (see Eq \eqref{eq:user-reward}). This is further confirmed on \fig{fig:reward_conv}{c}, where we guarantee long-term average constraint \eqref{eq:C1}.

\begin{figure}[!t]
    \centering
    \includegraphics[width=\columnwidth]{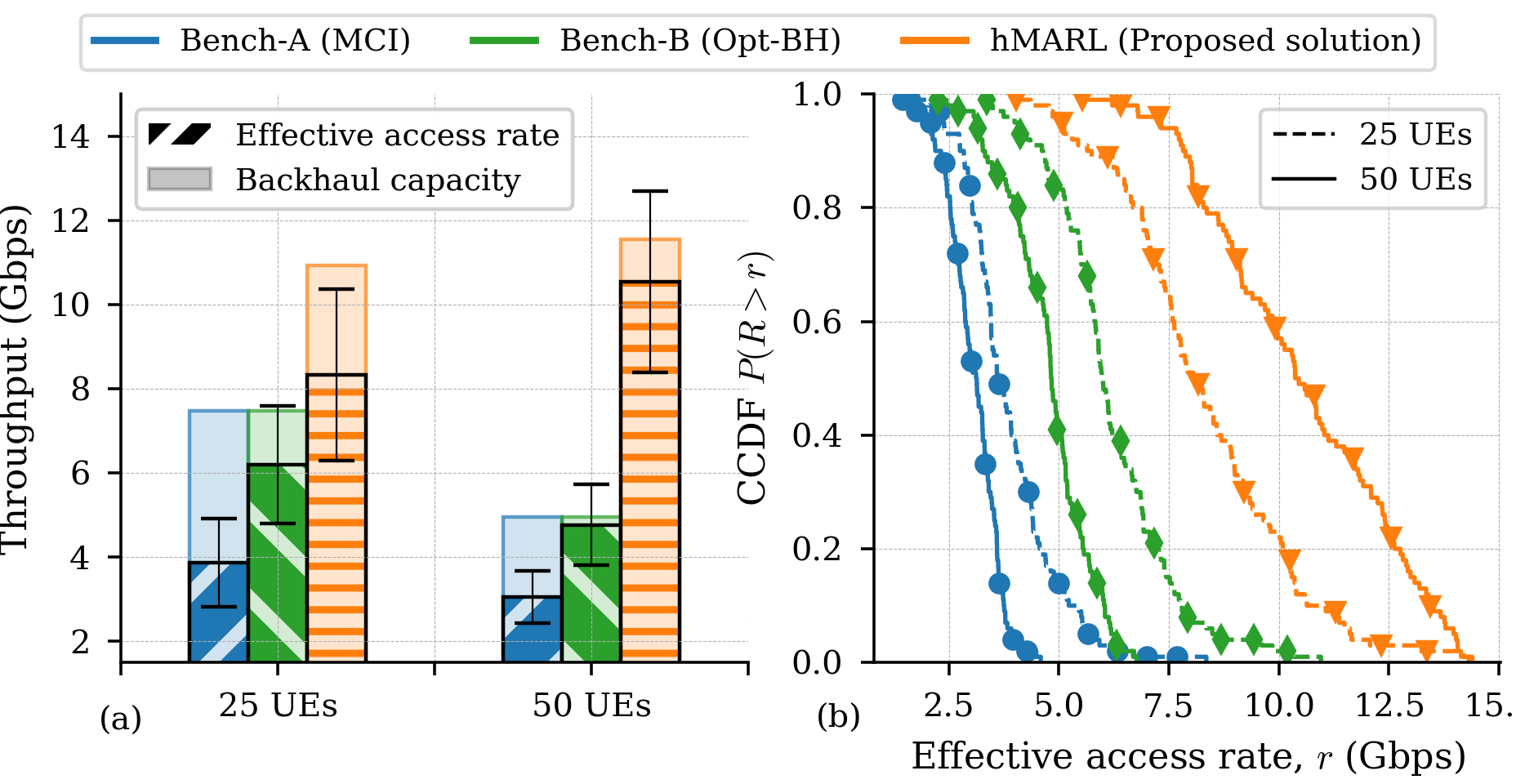}
    \caption{Performance comparison with baseline approaches. We plot the histograms and corresponding cumulative distributed functions (CCDF) over $500$ random Monte-Carlo deployments and channel realizations.}
    \label{fig:hRL_perf}
\end{figure}

\vspace{0.25cm}
\noindent
\textbf{Performance comparison with baseline approaches.} %

\fig{fig:hRL_perf} compares the performance of our proposed solution (referred to as \texttt{hMARL}) to two baselines employing \texttt{Max-SNR} algorithm for managing access links: it assigns each UE to the IAB station providing the maximum signal-to-noise ratio (SNR). The first benchmark (referred to as \texttt{Bench-A}) allocates backhaul capacity using the maximum carrier to interference (\texttt{MCI}) strategy, where the part of the capacity dedicated to a UE is proportional to its spectral efficiency \cite{Domenico2013BackHaul}. The second one (referred to as \texttt{Bench-B}) optimally allocates backhaul capacity by solving our proposed convex problem \eqref{eq:BH-BW}. All these benchmarks adopt a centralized exhaustive \texttt{Kmeans} clustering algorithm to position each mIAB node at the centroid of the closest cluster. In contrast, our proposed solution learns to autonomously and distributively position mIAB nodes to maximize user coverage and backhaul capacity while, at the same time, optimizing user association without the need for any central coordinator. {To show the effectiveness of our solution, we focus on the mIAB's network since UEs associated with the IAB donor do not rely on a backhaul link. Hence, we can observe in the histograms of \fig{fig:hRL_perf}{a} that our proposed solution clearly outperforms the two baseline solutions. It provides a $50\%$ additional backhaul capacity, increasing network sum-rate by $115\%$ and $34.4\%$ compared to \texttt{Bench-A} and \texttt{Bench-B}, respectively. This significant gain is also noticeable in \fig{fig:hRL_perf}{b}, which plots the \emph{service coverage} probability, \ien, the probability of having the sum-rate above a given threshold. When we set this threshold, \egn, to $6.5\Gbps$ for $K=25$ UEs, our proposed solution guarantees $84\%$ service coverage compared to $3\%$ and $40\%$ for \texttt{Bench-A} and \texttt{Bench-B}, respectively. In addition, though we perform the training for $K=25$ UEs, we also evaluate the performance for $K=50$ without any relearning procedure to show the capability of the proposed approach to cope with varying numbers of UEs and network topology. When the number of UEs increases, the performance of the baselines decreases accordingly. Indeed, the network sum-rate does not increase systematically with the increase in the number of UEs due to co-channel interference and limited backhaul capacity. However, in contrast to baseline solutions, our proposed solution exhibits additional performance improvement with up to $3.5\times$ and $2.2\times$ the sum-rate of \texttt{Bench-A} and \texttt{Bench-B}, respectively. This significant gain results from the ability of our algorithm to appropriately serve UEs given their traffic requests and to adapt the backhaul network accordingly.
}

\section{Conclusion}
This work investigated the problem of optimal mobile IAB nodes positioning, user association, and backhaul capacity allocation. We focus on \emph{in-band} relaying IAB network, which conducts the access and backhaul links on the same frequency band. Such network suffers from severe interference constraints limiting mostly backhaul capacity and the performance of end-users. {To address this problem, we propose a novel and scalable two-level hierarchical multi-agent learning mechanism where a high-level policy determines the positioning of mobile IAB nodes by jointly optimizing user coverage and backhaul capacity.} This procedure defines a goal for a low-level policy, which optimizes the user association and allocate the backhaul capacity. Once learned, the policies are distributed and can autonomously reconfigure the backhaul network \wrt environment dynamics and access network requirement. 
Numerical evaluations show the advantage of our solution, which provides up to $3.5\times$ network sum-rate increase compared to baseline approaches. Future work will exploit the results of this study to provide backhaul support to unmanned aerial vehicles network.

\bibliographystyle{ieeetr}
\bibliography{biblio}

\end{document}